\documentclass[12pt]{article}    % I prefer "article" over "amsart" for homework. "report" also looks nice

\usepackage{amssymb, amsmath, amsthm, mathtools} % math stuff
\usepackage{graphicx, dsfont, enumerate, color, setspace} % some general stuff
\usepackage{url}
\usepackage{caption}
\usepackage{subcaption}

\usepackage{authblk}

\usepackage{fullpage} % sets margins a bit smaller
\usepackage{natbib}
\bibpunct{(}{)}{;}{a}{,}{,}

\renewcommand{\H}{W}
\newcommand{\Hi}{W_{[i]}}

\newcommand{\Zobs}{Z^{\mathrm{obs}}}
\newcommand{\Tobs}{T^{\mathrm{obs}}}

\newcommand{\Zset}{\mathcal{Z}} %% set of assignments
\newcommand{\Uset}{\mathcal{U}} %% set of units
\newcommand{\Uall}{\mathbb{U}}
\newcommand{\C}{\mathcal{C}}

\newcommand{\Hset}{\mathcal{H}}

\newcommand{\Ueff}{\mathcal{U}^{\mathrm{eff}}}

\newcommand{\iv}{\mathbb{I}}

\newcommand\indep{\protect\mathpalette{\protect\independenT}{\perp}}
\def\independenT#1#2{\mathrel{\rlap{$#1#2$}\mkern2mu{#1#2}}}

\newcommand{\w}{m}
\newcommand{\wfg}{m_{[f]}}

\newcommand{\pr}{{\rm pr}}

\newcommand{\Yobs}{Y^{\mathrm{obs}}}

\newcommand{\eff}{\mathrm{eff}}

\newcommand{\pval}{\mathrm{pval}}
\newcommand{\UZ}{\mathbb{U}(Z)}

\newtheorem{theorem}{Theorem}[section]
\newtheorem*{theorem*}{Theorem}

\newtheorem{definition}{Definition}
\newtheorem{proposition}{Proposition}
\newtheorem*{proposition*}{Proposition}

\newtheorem{example}{Example}

\title{\bf Conditional randomization tests of causal effects with interference between units\thanks{Email: \texttt{afeller@berkeley.edu}. The authors thank Peng Ding, Dean Eckles, Michael Hudgens, Kosuke Imai, Luke Miratrix, Todd Rogers, Fredrik S{\"a}vje, John Ternovski, and Teppei Yamamoto for helpful feedback and discussion. AF also thanks the excellent research partners at the School District of Philadelphia, especially Adrienne Reitano and Tonya Wolford.}}
\author[1]{Guillaume Basse}
\affil[1]{UC Berkeley, Dept. of Statistics}
\author[2]{Avi Feller}
\affil[2]{UC Berkeley, Goldman School of Public Policy and Dept. of Statistics}
\author[3]{Panos Toulis}
\affil[3]{University of Chicago, Booth School of Business}

\date{\today}

\begin{document}

\maketitle

%\tableofcontents

\begin{abstract}
Many causal questions involve interactions between units, also known as interference, for example between individuals in households, students in schools, or firms in markets. 
In this paper we formalize the concept of a conditioning mechanism, which provides a framework for  constructing 
valid and powerful randomization tests under general forms of interference.
We describe our framework in the context of two-stage randomized designs and apply our approach to a randomized evaluation of an intervention targeting student absenteeism in the School District of Philadelphia. We show improvements over existing methods in terms of computational and statistical power.
\end{abstract}

\doublespacing

%%%%%%%%%%%%%%%%%%%%%%%%%%%%%%%%%%%%%%%%%%%%%%%%%%%%%%%%%%%%%%%%%%%%%%%%%%%

\section{Introduction}
% - - - - -
\label{section:introduction}
Classical approaches to causal inference assume that units do not interact with each other, known as the no-interference assumption~\citep{cox1958planning}. Many causal questions, however, are inherently about interference between units~\citep{sobel2006randomized, hudgens2012toward}, and
standard approaches often break down. 
For example, randomization tests on sharp null hypotheses of no effect~\citep{fisher1935design}
are more challenging 
in the presence of interference because these hypotheses are usually not sharp when there are interactions between units.

\citet{aronow2012general} and~\citet{athey2016exact} addressed this issue by proposing conditional randomization tests restricted to a subset of units, often called focal units, and a subset of assignments for which the specified null hypothesis is sharp for every focal unit. While the randomization-based approaches in these papers have advantages over  model-based approaches~\citep{bowers2013reasoning, toulis2013estimation},
they either explicitly forbid any conditioning that depends on the observed treatment assignment \citep{athey2016exact} or only give limited guidance on how to carry out such conditioning \citep{aronow2012general}.
This constraint may affect testing power because, under interference,  realized interactions between units depend on the treatment assignment. 
The constraint also makes implementing the procedure as a permutation test more difficult, which is an 
often-neglected practical problem.

In this paper we develop a framework for constructing valid and powerful randomization tests under interference.
To do so, we extend current approaches by formalizing the concept of a conditioning mechanism.
% Our key contribution is to introduce the concept of a conditioning mechanism. 
The proposed framework enables flexible conditional randomization tests that can condition on the observed treatment assignment.
We show that current methods for randomization tests in the presence of interference are special cases of our framework
and correspond to mechanisms that generally fail to leverage the problem structure effectively. 
For example, current methods often include units whose outcomes provide no information for the null hypothesis of interest, leading to unnecessary loss of power.
In our framework, it is straightforward to exclude such units from the test via additional conditioning. 
Furthermore, more flexible conditioning typically yields permutation tests that are straightforward to implement, resulting in computational gains.

We apply this approach to two-stage randomized designs, which are often used for assessing causal effects related to interference~\citep{hudgens2012toward}. First, we show how to apply our framework in this setting by 
suggesting concrete conditioning mechanisms for various hypotheses. Second, we analyze data from a two-stage randomized evaluation of an intervention targeting student absenteeism in the School District of Philadelphia. 
Our test is more powerful than alternative methods when applied to the absenteeism study, with a roughly one-third increase in statistical power.
Furthermore, our method yields a permutation test on the exposures of interest;
alternative methods cannot be implemented as permutation tests, instead requiring complicated adjustments.

% # # # # # # # # # # # # # # # # # # # # # # # # # # # # # # # # # # # # # # # # # # # 

\section{General Results for Randomization Testing}
\label{section:theory}
% - - - - -

\subsection{Classical randomization tests}
\label{section:background}
% - - 
Consider $N$ units indexed by $i = 1, \ldots, N$, and a binary treatment assignment vector 
$Z \in \{0,1\}^N$,
where the $i$-th component, $Z_i$, is the treatment assignment of unit $i$. The assignment vector is sampled with probability $\pr(Z)$. Denote by $Y_i(Z)$ the scalar potential 
outcome of unit $i$ under assignment vector $Z$. Under the stable unit 
treatment value assumption~\citep{rubin1980comment}, the potential outcome of unit $i$ depends only 
on its own assignment. Each unit therefore has two potential outcomes, typically denoted as $Y_i(1)$ and $Y_i(0)$, which correspond to outcomes when unit $i$ receives treatment or control, respectively. A classic goal is to test the sharp null hypothesis of zero treatment effect for all units,
\begin{equation}\label{eq:sharp-null}
	H_0 \, : \, Y_i(1) =Y_i(0) \quad (i = 1, \ldots, N).
\end{equation}

% Following a long tradition, 
We can assess $H_0$ by randomization~\citep{fisher1935design}. 
Let $T(Z \mid Y)$ denote the test statistic; for example, 
$T(Z \mid Y) = \mathrm{Ave}(Y_i \mid Z_i=1) - 
\mathrm{Ave}(Y_i \mid Z_i=0)$ 
is the usual difference in means between treated and 
control units, where $\mathrm{Ave}$ denotes sample average.
Let $\Tobs = T(\Zobs \mid \Yobs)$ 
denote the observed value of the test statistic, 
where $\Zobs\sim \pr(\Zobs)$ is the observed 
assignment vector in the experiment, and $\Yobs = Y(\Zobs)$ is the 
corresponding observed outcome vector. Finally, calculate the p-value 
\begin{equation}\label{eq:sharp-null-pval}\pval(\Zobs) = E_Z[\iv\{T(Z \mid \Yobs) \geq \Tobs\}],\end{equation}
where $\iv(\cdot)$ is the indicator function, and $E_Z$ is the 
expectation with respect to the distribution of $Z$. This test is 
valid at any level $\alpha$; that is, 
$\pr\{ \pval(\Zobs) \leq \alpha \} \leq \alpha$, for all 
$\alpha\in[0, 1]$ when the null hypothesis is true.
The key property that ensures validity of \eqref{eq:sharp-null-pval} is that, under $H_0$, the value of $T(Z \mid Y)$ can be imputed for every possible counterfactual assignment vector $Z'$, using only outcomes $\Yobs$ observed under $\Zobs$. 
This property allows us to construct the correct sampling distribution of the test statistic. We state the property formally in the following definition, as it will be useful for extending the classical randomization test to settings with interference.

\begin{definition}%[Imputable test statistic]
A test statistic $T(Z\mid Y)$ is imputable with respect to a null hypothesis $H_0$ if for all $Z, Z'$, for which for which $\pr(Z)>0$ and $\pr(Z')>0$,
\begin{equation}
\label{eq:T}
T\{Z' \mid Y(Z')\} = T\{Z' \mid Y(Z)\}.
\end{equation}
\end{definition}
The key property of an imputable test statistic is that we can simulate its sampling distribution under the null hypothesis $H_0$, even though we only observe one vector of outcomes, namely $Y(\Zobs)$. 
In the classical setting with no interference,~\eqref{eq:T} follows from the stable unit treatment value assumption and the sharp null hypothesis in~\eqref{eq:sharp-null}, which together imply 
that $Y(Z') = Y(Z)$, for any possible $Z, Z'$. Thus, in the classical setting all potential outcomes are imputable, and, by extension, any test statistic is imputable.

\subsection{Randomization tests via conditioning mechanisms}
We now demonstrate that we can obtain valid tests without requiring the stable unit treatment value assumption or a sharp null hypothesis. To do so, we introduce the concept of a conditioning event, $\C$, which is a random variable that is realized in the experiment; we leave this concept abstract for now and give concrete examples below.
The key idea is to choose an event space and some conditional distribution $\w(\C \mid Z)$ on that space, such that, conditional on $\C$,  a test statistic $T(Z\mid Y, \C)$ is  imputable with respect to the null hypothesis.  
We refer to 
 $\w(\C \mid Z)$ as the conditioning mechanism;
 % where $\w(\C|Z) > 0$ implies that  $Z\in\Zset$ and 
 $\w(\C \mid Z)$ and the design $\pr(Z)$ together induce a joint distribution, $\pr(Z, \C; \w) = \w(\C \mid Z) \pr(Z)$.
With these concepts, we can now state our first main result.
\begin{theorem}
\label{th:abstract}
Let $H_0$ be a null hypothesis and $T(Z \mid Y, \C)$
a test statistic, such that $T$ is imputable with respect to $H_0$ 
under some conditioning mechanism $\w(\C \mid Z)$;
that is, under $H_0$,	%
\begin{align}
\label{eq:T2}
	T\{Z' \mid Y(Z'), \C\} = T\{Z' \mid Y(Z), \C \},
	\end{align}
for all $Z, Z',\C$ for which $\pr(Z, \C; m) > 0$ and $\pr(Z', \C; m)> 0$. 
Consider the procedure where we first draw $\C \sim \w(\C \mid \Zobs)$, and then compute the conditional p-value,
\begin{align}
\label{eq:rand-distr}
\pval(\Zobs; \C) = E_Z[\iv\{T(Z \mid \Yobs, \C) > \Tobs\} \mid \C],
\end{align}
	where $\Tobs = T(\Zobs \mid \Yobs, \C)$, and the expectation is with respect to $\pr(Z \mid \C) = \pr(Z, \C; \w) / \pr(\C)$.
	This procedure
	is valid  at any 
	level, that is,
	$\pr\{\pval(\Zobs; \C) \leq \alpha \mid \C \} \leq \alpha$, 
    for any $\alpha\in [0, 1]$, under $H_0$.
\end{theorem}
Equation~\eqref{eq:T2} is the critical property that the test statistic is imputable, and directly generalizes~\eqref{eq:T}.
As before, the key implication of equation~\eqref{eq:T2} is that  we can simulate from the null distribution of $T\{Z \mid Y(Z), \C\}$, given any possible conditioning event $\C$.
Theorem~\ref{th:abstract} allows us to extend conditional randomization testing to more complicated settings, including testing under interference. 
Before turning to these settings, we briefly demonstrate that classical examples of randomization testing are special cases of Theorem~\ref{th:abstract}.

\begin{example}%[Classical Fisher randomization test]
	Let the conditioning event space be such that $T(Z \mid Y, \C) \equiv T(Z \mid Y)$ and $\C \indep Z$. Then the procedure in Theorem~\ref{th:abstract} reduces to the classical Fisher randomization test described in Section~\ref{section:background}. 
\end{example}

\begin{example}%[Correcting for covariate imbalance]
\cite{hennessy2016conditional} propose a conditional test that adjusts for covariate imbalance, quantified via a function $B(Z, X)$, where $X$ denotes a covariate vector. For instance, $B$ may be 
the vector of covariate means in each treatment arm, 
$B(Z, X) = \{\mathrm{Ave}(X_i | Z_i=1), 
\mathrm{Ave}(X_i | Z_i=0)\}$. Let $\pr(Z) = \text{Unif}\,\{(0,1)^N\}$ be a Bernoulli randomization design, and consider the conditioning mechanism
defined as $\pr(Z, \C) = \iv\{B(Z, X) = \C\} \pr(Z)$. Let 
$T(Z \mid Y, \C) \equiv T(Z \mid Y)$ be independent of $\C$, and let $H_0$ be as in~\eqref{eq:sharp-null}. Then the procedure of 
Theorem~\ref{th:abstract} corresponds exactly to that of \citet{hennessy2016conditional}. 
\end{example}

%%%%%%%%%%%%%%%%%%%%%%%%
\section{Randomization Tests for General Exposure Contrasts}
\label{section:testing-interference}
\subsection{General exposure contrasts}
\label{section:general-exposure}
We now turn to constructing valid randomization tests in the presence of interference.
Following~\citet{manski2013identification}
and \citet{aronow2013estimating}, we consider an 
exposure mapping $h_i(Z): \{0, 1\}^N\to \Hset$, where $\Hset$ is an arbitrary set of possible treatment 
exposures equipped with an equality relationship.
Given an exposure mapping, a natural assumption that generalizes the classical stable unit treatment value assumption is
\begin{equation}\label{a:sutva2}
	Y_i(Z) = Y_i(Z') \quad (i=1, \ldots, N)\;\text{for all}\;Z, Z'~\text{for which}~h_i(Z) = h_i(Z').
\end{equation}
This assumption states that potential outcomes are functions only of the exposure, rather than of the entire assignment vector. 
In the most restrictive case of no interference, the exposure mapping is $h_i(Z) = Z_i$;
in the most general case without any restrictions on interference, the exposure mapping is $h_i(Z) = Z$. 
An example of an intermediate case is if $h_i(Z) = \sum_{j\in\mathcal{N}_i} Z_j$, 
where $\mathcal{N}_i$ is the set of unit $i$'s neighbors in some network between units, and the exposure mapping of $i$ is therefore the number of $i$'s treated 
neighbors~\citep{toulis2013estimation}. In these examples, 
we implicitly defined $\Hset = \{0, 1\}$, $\Hset = \{0, 1\}^N$, and $\Hset = \mathbb{N}$, respectively.

We can now formulate hypothesis tests on  contrasts between treatment exposures. Let $\{a, b\} \subseteq \Hset$ be two exposures of interest. The null hypothesis on the  contrast between exposures $a$ and $b$ is
\begin{equation}
\label{eq:Ho_contrast}
	H_0:~Y_i(Z) = Y_i(Z') \quad (i=1,\ldots, N)\;\text{for all}~
	Z, Z'~\text{for which}~h_i(Z), h_i(Z') \in \{a,b\}.
\end{equation}
The classical sharp null hypothesis in~\eqref{eq:sharp-null} is a special case of~\eqref{eq:Ho_contrast}, with $\Hset=\{a, b\}=\{0, 1\}$. Under the no interference setting of~\eqref{eq:sharp-null}, we can permute the vector of unit exposures $\{a, b\}$ by permuting the treatment assignment vector because the null hypothesis contains all possible exposures.
In most interference settings, however, the null hypothesis in~\eqref{eq:Ho_contrast} is not sharp because it only considers a subset of possible exposures.
As a result, observing $Y(\Zobs)$ gives only limited information about counterfactual outcomes $Y(Z')$, with $Z' \neq \Zobs$. Since $h_i$ may have arbitrary form, we cannot permute unit exposures by naively permuting the treatment assignment vector.

\subsection{Constructing valid tests for general exposure contrasts}
\label{section:valid-test}
Testing a contrast hypothesis as in~\eqref{eq:Ho_contrast} is challenging because only a subset of units is exposed to exposures $a$ or $b$, and only for a subset of assignment vectors. We therefore construct conditioning events in terms of both units and treatment assignment vectors.
Specifically, let $\mathbb{C} = \{(\Uset, \Zset) : \Uset\subseteq\mathbb{U}, 
	\Zset\subseteq\mathbb{Z}\}$ be the space of conditioning events, where $\mathbb{U}$ denotes the power set of units, and $\mathbb{Z}$ 
	denotes the power set of assignment vectors.
For some conditioning event $\C = (\Uset, \Zset)\in\mathbb{C}$, the conditioning mechanism can be decomposed, without loss of generality, as
\begin{align}
\label{eq:deco}
m(\C \mid Z) = f(\Uset \mid Z) g(\Zset \mid \Uset, Z),
\end{align}
where $f$ and $g$ are distributions over $\mathbb{U}$ and $\mathbb{Z}$, respectively. 
Given conditioning event $\C=(\Uset, \Zset)$, 
we consider test statistics, $T(Z\mid Y, \C)$, that depend 
only on outcomes of units in $\Uset$; following terminology in~\cite{athey2016exact}, we call $\Uset$ the set of focal units.
For example, we can set $T(Z \mid Y, \C)$ to be
the difference in means between focal units exposed to $a$ and units exposed to $b$:
\begin{equation}
	\label{eq:T-contrast}
T(Z \mid Y, \C) =	\mathrm{Ave}\{Y_i \mid i\in\Uset, h_i(Z)=a\} - \mathrm{Ave}\{Y_i \mid i\in\Uset, h_i(Z)=b\}.
	\end{equation}
\begin{theorem}
\label{thm:concrete}
	Let $H_0$ be a null hypothesis as in~\eqref{eq:Ho_contrast}, let $\w(\C \mid Z)$ be a conditioning mechanism as in~\eqref{eq:deco}, let $\C = (\Uset, \Zset)$, and let $T$ be a test statistic defined only on focal units, as in~\eqref{eq:T-contrast}.
	Then, $T$ is imputable under $H_0$ if $\w(\C \mid Z) > 0$ implies that $Z\in\Zset$, 
    and for every $i\in\Uset$ and $Z'\in\Zset$, that
	\begin{equation}
\label{eq:c1}
h_i(Z') \in\{a, b\},~\,h_i(Z)\in\{a, b\},
\end{equation}
or
\begin{equation}
 \label{eq:c2}
h_i(Z') = h_i(Z),~\,h_i(Z)\notin\{a, b\}.
 \end{equation} 
If $T$ is imputable the randomization test for $H_0$ described in Theorem~\ref{th:abstract} is valid at any 
	level $\alpha$. 
\end{theorem}
Building on Theorem~\ref{thm:concrete}, we can construct a 
family of valid conditional randomization tests by enumerating 
the assignment vectors for which conditions~\eqref{eq:c1} and~\eqref{eq:c2} hold. 
% While the setup of Theorem~\ref{thm:concrete} is fairly general, we can use the above results to construct a family of valid conditional randomization tests.
As an example, for any choice of $f(\Uset \mid Z)$ we could define $g(\Zset \mid \Uset, Z)$ as follows:
\begin{align}
\label{eq:g}
g(\Zset \mid \Uset, Z) = 1,~\text{only if}~\Zset = \{Z'\in\mathbb{Z}: \text{Equations~\eqref{eq:c1} and~\eqref{eq:c2} are satisfied for}~Z'\}.
\end{align}
With this definition, $g$ is degenerate, and so the conditioning mechanism $m(\C \mid Z)$ is indexed solely by the conditional distribution, $f$, of focal units; we denote these conditioning mechanisms $\wfg$. Thus, our methodology provides many possible conditioning mechanisms that yield valid conditional randomization tests by construction. 
We can then select conditioning mechanisms with desired characteristics, such as high power. For example, we can choose $f$ to maximize the expected number of focal units whose outcomes are informative 
about $H_0$. We refer to this set of units as the set of effective focal units, $\eff(\Uset)$, where 
$\eff(\Uset) = \{i\in\Uset: h_i(\Zobs)= a~\text{or}~b\}$. Similarly, we could ensure that the number of possible randomizations is also large, and instead maximize the quantity $|\eff(\Uset)| |\Zset|$. Many choices are possible and should be tailored to the specific application.

%%%%%%%%%%%%%%%%%%%%
\section{Interference in two-stage randomized trials}
\label{section:two-stage}
\subsection{Two-stage randomized trials}
\label{sec:two_stage}
We now turn to the use of conditional randomization tests in two-stage 
randomized trials, which are used to assess spillovers between 
units~\citep{hudgens2012toward}. Specifically, we consider the setting of~\citet{basse2016analyzing}, in which $N$ units reside in $K$ households indexed by $k = 1, \ldots, K$.
In the first stage of the two-stage randomized trial, $K_1$ households are assigned to treatment, completely at 
random. In the second stage, one individual in each treated household is assigned to treatment, completely at random. 
As before, $Z_i\in\{0, 1\}$ is the assignment of unit $i$, and $Z=(Z_1, \ldots, Z_N)$ is the entire 
assignment vector. 
There is a residence index $R_{ij}$, such that $R_{ij}=1$ if unit $i$ resides in household $j$, and  is 0 otherwise. 
Let $[i] = \sum_j j R_{ij}$ denote the household wherein unit $i$ resides.
Finally, let $\H = (\H_1, \ldots, \H_K)$ denote the assignment vector on the household level, 
so that $\H_j  = \sum_{i} Z_i R_{ij}$.

The stable unit treatment assumption is not realistic in this context, so we make two assumptions on the 
interference structure that will imply a specific exposure mapping. First, we make the partial interference 
assumption \citep{sobel2006randomized}: units can interact within, but not between, households. 
Second, 
we make the stratified interference assumption \citep{hudgens2012toward}: unit $i$'s potential outcomes only depend 
on the number of units treated in the household, here 0 or 1, rather than the precise identity of the treated unit. 
~\citet{manski2013identification} calls this the anonymous interactions assumption. 
See~\citet{hudgens2012toward} for additional discussion.

These two assumptions can be expressed by the exposure mapping $h_i(Z) = (Z_i, \Hi)$. Since the potential outcome of unit $i$ depends only on $h_i(Z)$ by the assumption in equation~\eqref{a:sutva2}, for brevity we will use
$Y_i(Z_i, \Hi)$ to denote the value of $Y_i(Z)$.
Thus, unit $i$'s potential outcome can take only three values:
$$
Y_i(Z) \in \{Y_i(0, 0), Y_i(0, 1), Y_i(1, 1)\},
$$
that is, $Y_i(0,0)$ if unit $i$ is a control unit in a control household; $Y_i(0,1)$ if 
unit $i$ is a control unit in a treated household; and $Y_i(1,1)$ if unit $i$ is a treated unit in a treated household.
The fourth combination, $Y_i(0,1)$, is not possible because when unit $i$ is 
treated, household $[i]$ is also treated. Thus, the space of exposures is $\Hset = \{ a, b, c\}$,
with $a = (0,0), b = (0,1)$, and $c = (1,1)$. 

\subsection{A valid test for spillovers in two-stage designs}
\label{sec:spillovers}
We now focus on testing the null hypothesis of no spillover effect:
\begin{equation}
\label{eq:Ho_spill}
	H_0^s:~Y_i(0,0) = Y_i(0,1) \quad (i=1, \ldots, N).
	\end{equation}
The Supplementary Material contains analysis and results for the null hypothesis of no primary effect, $H_0^p: Y_i(0, 0) = Y_i(1, 1)$, for every unit $i$. 
Equation~\eqref{eq:Ho_spill} is a special case of the exposure contrast as defined in \eqref{eq:Ho_contrast}, with $a=(0, 0)$ and $b=(0, 1)$. As in Section~\ref{section:valid-test}, we set the test statistic to be the difference in means between the two exposures. The challenge is to find a conditioning mechanism that guarantees validity while preserving power.

We impose two constraints on our choice of focal units. First, units that are exposed to $c = (1,1)$  are excluded from being selected as focal units because these units do not contribute to the test statistic. Equivalently, we want to exclude units assigned to $Z_i = 1$ from being focal. This is therefore an example of conditioning using observed assignment $Z$, which avoids wasting units in the randomization test. 
Second, we choose a single non-treated unit at random from each household as the focal unit. In the Supplementary Material, we show that choosing one focal unit per household leads to a randomization test that is equivalent to a permutation test on the exposures of interest, $a=(0,0)$ and $b=(0, 1)$, which greatly simplifies computation.

\begin{proposition}\label{prop:two-stages-test}
Consider the following testing procedure:
	\begin{enumerate}
		\item in control households ($W_j = 0$), choose one unit at random. In treated households ($W_j = 1$), 
		choose one unit at random among the non-treated units ($Z_i = 0$); 
		\item compute the distribution of the test statistic in equation~\eqref{eq:T-contrast}  induced by 
		all permutations of exposures on the chosen units, using $a=(0, 0)$ and $b=(0, 1)$ as the 
		contrasted exposures;
		\item compute the p-value.
	\end{enumerate}
	Steps 1--3 define a valid procedure for testing the 
	null hypothesis of no spillover effect, $H_0^s$.
\end{proposition}
We show in the Supplementary Material that the procedure in Proposition~\ref{prop:two-stages-test} is an application of Theorem~\ref{thm:concrete} with a conditioning mechanism defined by%\avi{D hates this equation}
%
% \begin{equation}
% \label{eq:cond-f}
% 	f(\Uset \mid Z) = \text{Unif}\,\{\Uset \subseteq \Uall \,:\, 
% 	Z_i \iv(i \in \Uset) = 0,~\text{for every } i, \sum_i \iv(i \in \Uset) R_{ij} = 1, 
% 	\text{for every}~j\}.
% \end{equation}
\begin{equation}
\label{eq:cond-f}
	f(\Uset \mid Z) = \text{Unif}\,\left\{\Uset \subseteq \Uall \,:\, 
	Z_i \iv(i \in \Uset) = 0, \;\;\;\sum_{i'} \iv(i' \in \Uset) R_{i'j} = 1, 
	\;\;\text{for every}~i,j\right\}.
\end{equation}
As discussed earlier, the first constraint in~\eqref{eq:cond-f}
ensures that we only select focal units, $i\in\Uset$, that are not assigned to treatment; the second constraint restricts the focal set to one unit per household.

\subsection{Comparison with existing methods}
% - - 
Our approach builds on several existing methods. 
~\citet{aronow2012general},
who outlines some ideas that we discuss here, develops a test for the null hypothesis of no spillover effect.
Although that paper does not exclude conditioning on $Z$, it gives limited guidance on how such conditioning would work.
~\citet{athey2016exact} extends the method of \citet{aronow2012general} to a broader class of hypotheses,
but explicitly forbids the selection of focal units to depend on 
the realized assignment $Z$. 
In the Supplementary Material, we show that their approaches are equivalent to choosing a set of focal units independent of $Z$; that is, $f(\Uset \mid Z) \equiv f(\Uset)$. In fact, in the two-stage design we consider, the methods of~\citet{aronow2012general} and~\citet{athey2016exact} are identical; see the Supplementary Material.

\citet{athey2016exact} recognize that choosing focal units completely at 
random often yields tests with low power. They therefore propose more sophisticated approaches for selecting focal 
units using additional information. 
For instance, \citet{athey2016exact} advocate selecting focal units via $\epsilon$-nets: first select a focal unit, possibly at random, then choose 
subsequent focal units beyond a graph distance $\epsilon$ 
from that focal unit. In our applied example, this approach suggests choosing one focal unit at random from each 
household:
\begin{equation}
\label{eq:athey}
	f(\Uset \mid Z) \equiv f(\Uset) = \text{Unif} \, \left\{ \Uset \subseteq \Uall \,:\, \sum_i \iv(i \in \Uset) R_{ij} = 1,~\text{for every }j\right\}.
\end{equation}

Our proposed design in~\eqref{eq:cond-f} has two main advantages over the design in~\eqref{eq:athey}.   
First, we can implement our design via a simple permutation test, as described in Proposition~\ref{prop:two-stages-test}. This is not always possible for the design 
in~\eqref{eq:athey}. In fact, we show in the Supplementary Material that a conditioning mechanism based on~\eqref{eq:athey} is a permutation test only when households have equal size, which does not hold in our application.
In the absence of a permutation test, an analyst working with the conditioning mechanism defined by~\eqref{eq:athey} has to calculate the support of $g$ in~\eqref{eq:g} 
fully and exactly, and then take uniform draws over that set to sample from the correct randomization distribution. This calculation is exponentially hard. Moreover, there are no theoretical guarantees for when the test of~\citet{athey2016exact} can be implemented as a simple permutation test.

Second, unlike in our proposed design, the design in Equation~\eqref{eq:athey} may include treated units as focal units. Since treated units are not part of the effective focal set for testing the null hypothesis of no spillover effect, including them will reduce power. In particular, our design will always have at least as many effective focal units as the design in Equation~\eqref{eq:athey}, and at least as many assignment vectors in the randomization test.
To quantify this, suppose that all households have $n$ units.
We show in the Supplementary Material that for the
choice of $f(\Uset \mid Z) = f(\Uset)$ in~\eqref{eq:athey}, the number of effective focal units has distribution 
$|\text{eff}(\Uset)| \sim K - K_1 + \text{Binomial}(K_1, 1/n)$, where 
$K$ is the number of all households, and $K_1$ is the number of treated households, so $E\left\{|\text{eff}(\Uset)|\right\} = K - K_1(1 - 1/n)$. 
By contrast, the choice of $f(\Uset \mid Z)$ in 
Equation~\eqref{eq:cond-f} leads to a number of effective focal units that is always equal to $K$, 
the number of all households. 
For instance, in the experiment we describe next, there are 3,169 households with $n=2$ units. Restricting to this subset, the design in Equation~\eqref{eq:athey} has an average of 2,123 effective focal units, a reduction of one-third from our proposed design. 
%

%% PTA: 8/21
Finally, \citet{rigdon2015exact} propose a method for calculating exact confidence intervals in two-stage randomized designs with binary outcomes. However, it is not applicable to our setting with continuous outcome, nor is the proposed approximation well-suited for tests of a given null hypothesis.

\subsection{Application to a school attendance experiment}
\label{section:application}
We illustrate our approach using a randomized trial of an intervention designed to increase student attendance in the School District of Philadelphia~\citep{rogers2016intervening}. Following the setup in~\citet{basse2016analyzing}, we focus on a subset of this experiment with $N = 8,654$ students in $K = 3,876$ multi-student households, of which 
$K_1= 2,568$ were treated. For this subset, the district sent targeted attendance information to the parents about only one randomly chosen student in that household. The outcome of interest is
the number of days absent during the remainder of the school year. Following~\citet{rosenbaum2002covariance}, we focus on regression-adjusted outcomes, adjusting for a vector of pre-treatment covariates, including demographics and prior year attendance.
Additional details on the analysis are included in the Supplementary Material, including results for the primary effect.

To assess spillovers, we sample 100 sets 
$\Uset^{(l)}$ ($l = 1, \ldots, 100$) for both ours and 
\cite{athey2016exact}'s choice of function $f(\Uset \mid Z)$. 
For each set, we compute $p$-values for the null hypothesis of no spillover effect $H_0^s$ in Equation~\eqref{eq:Ho_spill} and report whether it rejects with $p < 0.05$. Overall, the test using \cite{athey2016exact}'s method rejects the null hypothesis of no effect for $66\%$ of focal sets; the test using our method rejects the null hypothesis of no effect for $92\%$ of focal sets. 

We also obtain confidence intervals and Hodges--Lehmann point estimates by inverting a sequence of randomization tests under an additive treatment effect model, 
$Y_i(1,0) = Y_i(0,0) + \tau^s$ ($i=1, \ldots, N$). For each focal set, obtaining these quantities is straightforward given $\Uset$ via standard methods~\citep{rosenbaum2002covariance}. Aggregating information across focal sets, however, remains an open problem; we discuss this briefly in Section~\ref{section:discussion}. 
For simplicity, we summarize the results by presenting medians across focal sets. 
For our proposed approach, the median value of the Hodges--Lehmann point estimates  is $\widehat{\tau}^{(\mathrm{cond})} \approx -1$ day, with 95\% confidence 
interval  $[-1.70, -0.34]$. For the method of \citet{athey2016exact}, the 
median estimate is $\widehat{\tau}^{(\mathrm{rand})} \approx -1.1$ days, with  associated 95\%
confidence interval $[-1.84, -0.28]$. Across focal sets, the average width of the confidence intervals obtained via 
\cite{athey2016exact}'s method is 1.60, compared to 1.42 with our approach, a reduction of 11\%.

Results from both approaches are in line with those obtained 
by \cite{basse2016analyzing} via unbiased estimators. These confirm the presence of substantial within-household spillover effect that is nearly as large as the primary effect, suggesting that intra-household dynamics play a critical role in reducing student absenteeism and should be an important consideration in designing future interventions.

% # # # # # # # # # # # # # # # # # # # # # # # # # # # # # # # # # # # # # # # # # # # 
\section{Discussion}
\label{section:discussion}
% PTA: 8/21
Constructing appropriate conditioning mechanisms can be challenging in settings more complex than two-stage designs. Doing so requires understanding the interference structure and finding powerful conditioning mechanisms subject to that structure. 
Furthermore,
conditioning mechanisms produce a distribution of p-values
across random choices for the conditioning event. While this does not affect the validity of the test, it raises problems such as 
 interpretation and sensitivity of the test results~\citep{geyer2005fuzzy}. At the same time, the distribution itself may contain information useful to improve the power of the test. In ongoing research, we are working to use multiple testing methods to address this problem.

%Finally, these tests necessarily focus on sharp null hypotheses. Going forward, we hope to assess the robustness of this approach to non-sharp nulls. 

% # # # # # # # # # # # # # # # # # # # # # # # # # # # # # # # # # # # # # # # # # # # 

%\section*{Acknowledgement}
%The authors thank Peng Ding, Dean Eckles, Michael Hudgens, Kosuke Imai, Luke Miratrix, Todd Rogers, Fredrik S{\"a}vje, John Ternovski, and Teppei Yamamoto for helpful feedback and discussion. AF also thanks the excellent research partners at the School District of Philadelphia, especially Adrienne Reitano and Tonya Wolford. 

% # # # # # # # # # # # # # # # # # # # # # # # # # # # # # # # # # # # # # # # # # # # 
%\section*{Supplementary material}
%
%Supplementary material available at {\it Biometrika} online includes the proofs of all theorems and claims, and additional empirical work from Section~\ref{section:application}. It also provides a counterpart to Section~\ref{section:spillover_main} for testing $H_0^p$.

\singlespacing
\bibliographystyle{chicago}
\bibliography{ref}

\clearpage
\appendix

\doublespacing
%%%%%%%%%%%%
\section{Proofs of theorems and statements}
%%%%%%%%%%%%

\subsection{Proof of validity of classical Fisher test}
%%%%%%%%%%%%

We reproduce the proof of~\cite{hennessy2016conditional} with
slight modifications. This proof will provide an introduction to the 
proof of the validity of the conditional test that follows.

% - -

\begin{proof}
We need to show that:
\begin{equation*}
	\pr( p \leq \alpha \mid H_0 ) \leq \alpha,~\text{for all}~\alpha \in [0,1],
\end{equation*}
where the probability is with respect to $\pr(\Zobs)$, and $p = \pval(\Zobs)$ is defined as
\begin{equation*}
	p = \pr\{T(Z \mid \Yobs) \geq T(\Zobs\mid \Yobs)\}.
\end{equation*}

Let $U$ be a random variable with the same distribution as $T(Z \mid \Yobs)$, 
as induced by $\pr(Z)$ and let $F_\mathrm{U}$ be its cumulative distribution 
function. We can then write
\begin{equation*}
	p = 1 - F_{\mathrm{U}}\{ T(\Zobs \mid \Yobs) \}.
\end{equation*}
By definition, under $H_0$ we have $Y(Z) = Y(\Zobs)$ for all $Z$, and so 
$T(Z \mid \Yobs) = T\{Z \mid Y(Z)\}$. It follows that, under $H_0$, $U$ 
has the same distribution as $T(Z \mid \Yobs)$. The randomness in 
$T(\Zobs\mid \Yobs)$ is induced by the randomness in $\Zobs$. In the 
testing procedure, $\Zobs \sim \pr(\Zobs)$. Combining with the above, 
we see that the distribution of $T(\Zobs\mid \Yobs)$ induced by 
$\pr(\Zobs)$ is the same as that of $U$ under $H_0$. We thus have
\begin{equation*}
	p = 1- F_{\mathrm{U}}(U).
\end{equation*}
By the probability integral transform theorem, $p$ is uniform, and so 
$\pr(p \leq \alpha \mid H_0) \leq \alpha$.
\end{proof}

\subsection{Proof of Theorem~\ref{th:abstract}}
%%%%%%%%%%%%

The proof of Theorem~\ref{th:abstract} follows that of the classical 
Fisher test, with some important modifications.
\begin{theorem*}[Theorem~\ref{th:abstract}]
%\label{th:abstract}
Let $H_0$ be a null hypothesis and $T(Z \mid Y, \C)$
a test statistic, such that $T$ is imputable with respect to $H_0$ 
under some conditioning mechanism $\w(\C \mid Z)$;
that is, under $H_0$, it holds that	%
\begin{align}
\label{eq:T2}
	T\{Z' \mid Y(Z'), \C\} = T\{Z' \mid Y(Z), \C \},
	\end{align}
for all $Z, Z',\C$, for which $\pr(Z, \C; m) > 0$ and $\pr(Z', \C; m)> 0$. 
Consider the procedure where we first draw $\C \sim \w(\C \mid \Zobs)$, and then compute the conditional p-value,
\begin{align}
\label{eq:rand-distr}
\pval(\Zobs; \C) = E_Z[\iv\{T(Z \mid \Yobs, \C) > \Tobs\} \mid \C],
\end{align}
	where $\Tobs = T(\Zobs \mid \Yobs, \C)$, and the expectation is with respect to $\pr(Z \mid \C) = \pr(Z, \C; \w) / \pr(\C)$.
	This procedure
	is valid  at any 
	level, that is,
	$\pr\{\pval(\Zobs; \C) \leq \alpha \mid \C \} \leq \alpha$,  
    for any $\alpha\in [0, 1]$.
\end{theorem*}

% - -

\begin{proof}

We need to show that
\begin{equation*}
	\pr(p_\C \leq \alpha \mid H_0, \C) \leq \alpha
\end{equation*}
for all $\C$ such that $\pr(\C \mid \Zobs) > 0$, where the probability is 
with respect to $\pr(\Zobs \mid \C)$, and $p_\C$ is defined as
\begin{equation*}
	p_\C = \pr\{T(Z \mid \Yobs, \C) \geq T(\Zobs \mid \Yobs, \C) \mid \C\}.
\end{equation*}
Fix $\C$. Let $U$ be a random variable with the same distribution as 
$T(Z \mid \Yobs, \C)$ as induced by $\pr(Z \mid \C)$ and let $F_\mathrm{U}$ be 
its cumulative distribution function. We can then write:
\begin{equation*}
	p_\C = 1 - F_{\mathrm{U}}\{T(\Zobs \mid \Yobs, \C)\}.
\end{equation*}
In the procedure, we have $\Zobs \sim \pr(\Zobs)$ and 
$\C \sim \pr(\C \mid \Zobs)$, implying that $\pr(\Zobs, \C) > 0$. 
So, by imputatability of the test statistic in 
Equation~\eqref{eq:T2} under $H_0$, 
\begin{equation*}
T\{Z \mid Y(Z), \C\} = T(Z \mid \Yobs, \C)
\end{equation*}
for all $Z \sim \pr(Z \mid \C)$, since this guarantees $\pr(Z, \C) > 0$.
This means that under $H_0$, $U$ has the same distribution as 
$T(Z \mid \Yobs, \C)$.
The randomness in $T(\Zobs \mid \Yobs, \C)$ is induced by the randomness
in $\Zobs$ conditional on $\C$. Combining with the above, we see that the
distribution of $T(\Zobs \mid \Yobs, \C)$ induced by $\pr(\Zobs \mid \C)$ 
is the same as that of $U$ under $H_0$. We thus have:
\begin{equation*}
	p_\C = 1 - F_{\mathrm{U}}(U).
\end{equation*}
By the probability integral transform theorem, $p_\C$ is uniform and so $\pr(p_\C \leq \alpha \mid H_0, \C) \leq \alpha$.
\end{proof}

\subsection{Proof of Theorem~\ref{thm:concrete}}
%%%%%%%%%%%%

For the reader's convenience we repeat the definitions of the contrast null hypothesis, conditioning mechanism, and test statistic, which are used in Theorem~\ref{thm:concrete}:
\begin{align}
\label{eq:Ho_contrast}
	& H_0:~Y_i(Z) = Y_i(Z'), i=1,\ldots, N,~\text{for all}~
	Z, Z'~\text{for which}~h_i(Z), h_i(Z') \in \{a,b\},
    \\
 \label{eq:deco}
& m(\C \mid Z) = f(\Uset \mid Z) g(\Zset \mid \Uset, Z),\\
\label{eq:T-contrast}
& T(Z \mid Y, \C) =	\textsc{Ave}\{Y_i \mid i\in\Uset, h_i(Z)=a\} - \textsc{Ave}\{Y_i \mid i\in\Uset, h_i(Z)=b\},
\end{align}
where $\C = (\Uset, \Zset)$, and $\Uset, \Zset$ are 
any subsets of units and assignment vectors, respectively and $\textsc{Ave}$ denotes the average.
The main challenge is to prove that the conditions of the theorem 
ensure that the test statistic in Equation~\eqref{eq:T-contrast} is 
imputable under $H_0$. 
\begin{theorem*}[Theorem~\ref{thm:concrete}]
%\label{thm:concrete}
	Let $H_0$ be a null hypothesis as in Equation~\eqref{eq:Ho_contrast}, $\w(\C \mid Z)$ be a conditioning mechanism as in Equation~\eqref{eq:deco}, and $T$ be a test statistic defined only on focal units, as in Equation~\eqref{eq:T-contrast}.
	Then, $T$ is imputable under $H_0$ if $\w(\C \mid Z) > 0$ implies that $Z\in\Zset$, 
    and for every $i\in\Uset$ and $Z'\in\Zset$ that
	\begin{align}
\label{eq:c1}
h_i(Z') & \in\{a, b\},~\text{if}~h_i(Z)\in\{a, b\},
\\
 \label{eq:c2}
 h_i(Z') & =h_i(Z),~\text{if}~h_i(Z)\notin\{a, b\}.
 \end{align} 
If $T$ is imputable the randomization test for $H_0$ as described in Theorem~\ref{th:abstract} is valid at any 
	level $\alpha$. 
\end{theorem*}

% - - - 

\begin{proof}
For a conditioning event $\C = (\Uset, \Zset)$, 
suppose that $\w(\C \mid Z) > 0$ implies that $Z \in \Zset$ and that:
\begin{equation*}
	\text{for all } i \in \Uset,  Z' \in \Zset, \quad 
    \begin{cases}
    	h_i(Z') \in \{a,b\} &\text{ if }\, h_i(Z) \in \{a,b\},\\
        h_i(Z') = h_i(Z) &\text{ if }\, h_i(Z) \not\in \{a,b\}.
    \end{cases}
\end{equation*}
Now let $Z, Z', \C$ be such that $\pr(Z, \C; m) > 0$ and 
$\pr(Z', \C; m) > 0$. By definition of a conditioning mechanism, this 
implies that $\w(\C \mid Z) > 0$ and $\w(C \mid Z') > 0$. It follows that 
$Z \in \Zset$ and $Z' \in \Zset$. Now take $i \in \Uset$. 
If $h_i(Z') \not\in \{a,b\}$, then, by assumption, $h_i(Z) = h_i(Z')$ 
since $Z, Z' \in \Zset$. And so by Equation~(5) of the main paper, we have 
that $Y_i(Z') = Y_i(Z)$. If instead $h_i(Z') \in \{a, b\}$, then 
$h_i(Z) \in \{a,b\}$ and so under the null hypothesis $Y_i(Z') = Y_i(Z)$, 
as well. Therefore, we proved that $Y_{\Uset}(Z') = Y_{\Uset}(Z)$, where 
$Y_{\Uset}(Z)$ denotes the subvector of outcomes of units in $\Uset$ under 
assignment vector $Z$. 
Since the test statistic, $T(Z \mid Y, \C)$, is defined only on $Y_\Uset$,
the subvector of outcomes of units in $\Uset$, it follows that 
$T\{Z' \mid Y(Z'), \C\} = T\{Z' \mid Y(Z), \C\}$, and so $T$ is imputable.
\end{proof}

\subsection{Proof of Proposition~1}
%%%%%%%%%%%%

\begin{proposition*}
Consider the following testing procedure:
	\begin{enumerate}
		\item In control households ($W_j = 0$), choose one unit at random. In treated households ($W_j = 1$), 
		choose one unit at random among the non-treated units ($Z_i = 0$). 
		\item Compute the distribution of the test statistic in 
        Equation~\eqref{eq:T-contrast}  induced by all permutations of 
        exposures on the chosen focal units, using $a=(0, 0)$ and 
        $b=(0, 1)$ as the contrasted exposures.
		\item Compute the p-value.
	\end{enumerate}
	Steps 1-3 outline a procedure that is valid for testing the 
	null hypothesis of no spillover effect, $H_0^s$.
\end{proposition*}

\begin{proof}
Define
\begin{equation*}
    \UZ = \{\Uset \in \Uall \,:\,Z_i \iv(i \in \Uset) = 0,
            ~i=1, \ldots, N,~\text{and}~\sum_i \iv(i \in \Uset) R_{ij} = 1, 
			\text{for every household}~j\}.
\end{equation*}
In words, $\UZ$ is the set of all subsets of units for which no unit in 
the subset is treated under $Z$, and each household has exactly one unit 
in the subset.
Step 1 of the procedure in Proposition~1 chooses focals according to 
conditioning mechanism 
$m(\C \mid Z) = f(\Uset \mid Z) g(\Zset \mid \Uset,Z)$, where we define
\begin{align}
    f(\Uset\mid Z) & = \text{Unif}\{\UZ \}, \\
      \label{eq:def_g} 
      g(\Zset\mid \Uset, Z) & = \iv[\Zset = \{Z' : 
       h_i(Z') \in\{(0, 0), (0, 1)\}
       \text{ for all}~i\in\Uset
       \}].
\end{align}
That is, $f(\Uset\mid Z)$ is uniform on $\UZ$ and $g$ is degenerate on the 
set of assignments for which all units in $\Uset$ are either in control or 
exposed to spillovers. In what follows, we fix a conditioning event 
$\C = (\Uset, \Zset)$.

Let $H = H(Z)\in\{0, 1\}^K$ denote the exposure of focal units under 
$Z$, where we use $0$ for control and $1$ for spillovers. Also, let 
$W = W(Z)\in\{0, 1\}^K$ denote the household assignment under 
assignment vector $Z$. Since there is one focal per household and 
household assignment determines the exposure of a focal, $H$ and $W$ 
are equal almost surely:
\begin{equation*}
H(Z) = W(Z), \text{ for all}~Z,\text{ and so we can write}
~H = W, \text{ almost surely}.
\end{equation*}
For any $Z, Z'\in\Zset$, it holds that 
\begin{equation*}
g(\Zset \mid\Uset, Z) = g(\Zset\mid\Uset, Z').
\end{equation*}
This follows from definition of $g$ in Equation~\eqref{eq:def_g} 
since $g(\Zset \mid \Uset, Z)\equiv g(\Zset \mid \Uset)$ does not depend 
on $Z$ given a fixed $\Uset$; note that $\Uset$ depends on $Z$ itself, 
but still $g$ does not depend on $Z$ if $\Uset$ is given.

For any $w\in\{0, 1\}^K$, it holds that:
\begin{equation*}
\sum_{Z : W(Z) = w} f(\Uset \mid Z) \pr(Z \mid W=w) = \mathrm{const.}
\end{equation*}
To see this, first note that $\pr(Z \mid W) = \prod_{k: W_k=1}^K 1 / n_k$, 
where $n_k$ is the number of units in the household. Furthermore, 
\begin{equation*}
	\sum_{Z:W(Z)=w} f(\Uset \mid Z) = \prod_{k : W_k=0} 1/n_k.
\end{equation*}
Therefore, 
\begin{equation*}
\sum_{Z : W(Z) = w} f(\Uset \mid Z) \pr(Z \mid W=w)
= \prod_k 1/n_k = \mathrm{const}.
\end{equation*}
Actually this is equal to the marginal probability of the focal set,
$\pr(\Uset)$.

We now put things together and prove that the conditioning mechanism 
yields a randomization distribution that is uniform in its support. 
Fix a conditioning event $\C = (\Uset, \Zset)$. Then,
\begin{align}
\pr(H \mid \C)
& = \pr(W \mid \C)
~[\textit{\small from Step 1}]\nonumber\\
& \propto \pr(\C \mid W) \pr(W)\nonumber\\
& \propto 
\sum_{Z} \pr(\C, Z \mid W) \pr(W)\nonumber\\
& \propto 
\sum_{Z:W(Z)=W} \pr(\C \mid Z) 
\pr(Z \mid W)\pr(W)\nonumber\\
& \propto 
\sum_{Z:W(Z)=W} f(\Uset \mid Z) g(\Zset \mid \Uset, Z)
\pr(Z \mid W)\pr(W)\\
& \propto 
g(\Zset \mid \Uset) \pr(W) \sum_{Z:W(Z)=W} f(\Uset \mid Z) \pr(Z \mid W)
\nonumber\\
& \propto 
\pr(W)
\nonumber\\
& = {N\choose N_1}^{-1}.
\end{align}
From the definition of the test statistic:
\begin{equation*}
T(Z \mid Y, \C) = T(Z' \mid Y, \C) \quad \text{ if } \quad H(Z) = H(Z').
\end{equation*}
Therefore, we can write $T(Z\mid Y, \C) \equiv T(H\mid Y, \C)$. From the 
above, we know that the conditional distribution of the focals' exposure 
under the particular conditioning mechanism is a permutation of their 
exposures under $\Zobs$, as prescribed by the testing procedure of 
Proposition~1. This is sufficient for validity since the test statistic is in fact a function of $H$. % as we just saw.

%%%%%%%%%%%%%%%%%%%%%%%%%%%%%%%%%%%%%%%%%%%%%%%%%%%%%%%%%%%%%%%%%%%%%%

%%%%%%%%%%%%
\section{Additional discussion of alternative methods}
%%%%%%%%%%%%

\subsection{Equivalence of tests from~\cite{athey2016exact} and~\cite{aronow2012general} for two-stage designs}
%%%%%%%%%%%%

The tests described by~\citet{athey2016exact} and~\citet{aronow2012general}
coincide for testing spillover effects, $H_0^s$, in our two-stage randomized setting. We will show that the method of \cite{aronow2012general} is equivalent to our procedure, with $f(\Uset \mid Z) = f(\Uset)$.
Briefly, the method of \cite{aronow2012general} can be summarized
as follows:
\begin{enumerate}
	\item Draw a set of units $\Uset\subset \Uall$, uniformly at random, as 
    in \cite{athey2016exact}.
	\item Compute the p-value by using the 
	conditional randomization distribution 
	 $\pr(Z \mid \Uset,  Z_\Uset = \Zobs_\Uset)$, where $Z_\Uset$ is the 
     subvector of $Z$ that is restricted to the units in $\Uset$.
\end{enumerate}
The conditional randomization distribution is therefore equal to:
\begin{flalign*}
	\pr(Z \mid \Uset, Z_\Uset = \Zobs_\Uset) &\propto \pr(\Uset, Z_\Uset = \Zobs_\Uset \mid Z) \pr(Z) 
	\propto \pr(Z_\Uset = \Zobs_\Uset \mid \Uset, Z) 
    \pr(\Uset \mid Z) \pr(Z) \\
	&=  \iv(Z_\Uset = \Zobs_\Uset)  \pr(\Uset) \pr(Z).
\end{flalign*}
Now, consider a conditioning event $\C = (\Uset, \Zset)$ from a mechanism 
$m_f(\C \mid Z) = f(\Uset) g(\Zset \mid \Uset, Z)$, where according to 
Equation~(11) in the main paper is degenerate on the set:
\begin{align}
\label{eq:zset}
\Zset= [Z': h_i(Z') = (1,1) \text{ if } h_i(\Zobs) = (1,1) \text{ and } h_i(Z') \in \{(0,0), (0,1)\} \text{ otherwise,  for all } i \in \Uset].
\end{align}
Under this definition and the setting of spillover effects, for every unit $i\in\Uset$ in the focal set and every assignment vector $Z'\in\Zset$ in the test,
we will have either $Z_i'=0$ if $\Zobs_i=0$ or 
$Z_i'=1$ if $\Zobs_i=1$. Thus, if $Z'\in\Zset$ it follows that $Z'_{\Uset} = \Zobs_{\Uset}$. 
Suppose the reverse is true, that is, $Z'_{\Uset} = \Zobs_{\Uset}$. 
Consider unit $i$ in the focal set for which $\Zobs_i = 1$. 
Then, $Z'_i=1$ as well, and so $h_i(Z') = h_i(\Zobs)$ for such units. 
Consider unit $i$ in the focal set for which $\Zobs_i = 0$. 
Then, $Z'_i=0$ as well, and so $h_i(Z') = \in\{(0, 0), (0, 1)\}$, 
by definition of exposures. Thus, if $Z'_{\Uset} = \Zobs_{\Uset}$ 
it follows that $Z'\in\Zset$. 
Therefore, the two statements are equivalent, and the conditioning mechanism with $f(\Uset \mid Z) = f(\Uset)$ will yield the same test as in~\citet{athey2016exact} and~\citet{aronow2012general}.

\subsection{When the test of Athey et al. (2017) is a permutation test}
%%%%%%%%%%%%%%%%%%%%%%%%%%%%%%%%%%%%%%%%%%%%%%%%%%%

The method of Athey et al. can be cast in our framework, where 
$f(\Uset \mid Z) = f(\Uset)$ , i.e., the selection of focals does not 
depend on the observed assignment, and where the randomization 
distribution, $\pr(Z \mid \C)$, is uniform over the set $\Zset$ defined in 
Equation~\eqref{eq:zset}.
We denote by $\Ueff(Z) = \{i \in \Uset: Z_i = 0\}$. We denote by 
$H_i = h_i(Z)$ the exposure of unit $i$ under assignment vector $Z$.

\bigskip 

First, notice that $\Ueff(Z) = \Ueff(\Zobs)$, for every $Z\in\Zset$. Now, 
consider unit $i \in \Ueff(\Zobs)$. We have:
\begin{flalign*}
	\pr(H_i = (1,0) \mid Z \in \Zset) 
    &= \pr(Z_i =0, W_{[i]}=1 \mid Z \in \Zset) \\
    &= \frac{\pr(Z_i =0 \mid W_{[i]} = 1)
    \pr(W_{[i]} = 1)}{\pr(Z \in \Zset)} \\
    &= \frac{(n_i-1)/n_i  {N \choose N_1}}{\pr(Z\in \Zset)} \\
   &  \propto \frac{n_i-1}{n_i}.
\end{flalign*}
We thus have the constraint that for all $Z \in \Zset$:
\begin{flalign*}
	\sum_{i \in \Ueff(Z)} \iv\{H_i(Z) = (1,0)\}
    &= N_1^{\text{eff}}(\Zobs).
\end{flalign*}
In words, the number of exposed units is constant for all $Z \in \Zset$. 
Putting it all together, we see that $P(H \mid Z \in \Zset)$ is such that:
\begin{enumerate}
	\item $\sum_{i \in \Ueff} \iv\{H_i = (1,0)\} = N_1^{\text{eff}}$.
    \item For all $i \in \Ueff$, $\pr\{H_i = (1,0) \mid Z \in \Zset(\Uset, \Zobs)\} \propto (n_i-1)/n_i$.
\end{enumerate}
\end{proof}
   
This result implies that the method of~\citet{athey2016exact} can be 
implemented as a permutation test only when the households are of equal 
sizes. This is not true in our application, and not expected to be true 
more generally, and thus poses computational challenges in implementing 
the test of~\cite{athey2016exact}.

% % % % % % % % % % % % % % % % %
\section{Simulations and analysis details}  %
 % % % % % % % % % % % % % % % % %

\subsection{Simulations}
\label{section:simulations}
% - - -

We compare the power of the test we proposed in the previous section, which chooses the focal units
conditionally on $\Zobs$, to that of the test in \cite{athey2016exact} which chooses the focals 
unconditionally of $\Zobs$. 
%  PANOS: Can we produce a figure with "unconditional focals" 
% instead of "random focals". Then change the text below.
We use the term ``unconditional focals'' to describe that approach, but we note that this could encompass selection of focals based on existing covariate information, such as a network between units. For example, \citet{athey2016exact} propose an approach where after a unit is selected as focal subsequent focal units are selected beyond a certain distance to the initial focal unit; this is known as the $\epsilon$-net approach.
Such approaches are still unconditional to the observed treatment assignment, $\Zobs$.
%%% > END PANOS.

Figure~\ref{fig:extreme} illustrates the potential power gains, by
considering the extreme case of $K=500$ households of equal size $n=50$ with $K_1 = 250$ treated
households, and focusing on the 
power of the test of no primary effect $H_0^{p}$. 
If we are interested in testing the no spillover effect hypothesis $H_0^s$, the expected difference in the 
number of effective focal units between our test and the test of \cite{athey2016exact} decreases with $n$. 
In the case of the no primary effect hypothesis $H_0^p$, the difference increases with $n$. This 
phenomenon is illustrated in Figure~\ref{fig:power-by-nis}.
%If we let $E(K^{eff}_r)$ be the expected number of
%effective focals in the random focal test, and $E(K^{eff}_c)$ be the expected number of focals in the 
%conditional focal test, it is easy to show that 
%%
%\begin{equation}\label{eq:primary}
%	E(K^{eff}_c - K^{eff}_r) = K_1(1-1/n)
%\end{equation}
%In other other
%words, for a given number of treated units $K_1$, the difference in number of effective focals, and so 
%the difference in power, increases with the size of the households. In the case of Figure~\ref{fig:extreme},
%the difference is $E(N^{eff}_c - N^{eff}_f) = 245$. This means that the conditional focals test has almost
%twice as many focal units as the random focals test.

\begin{figure}[b!]
	\centering
	\includegraphics[scale=0.45]{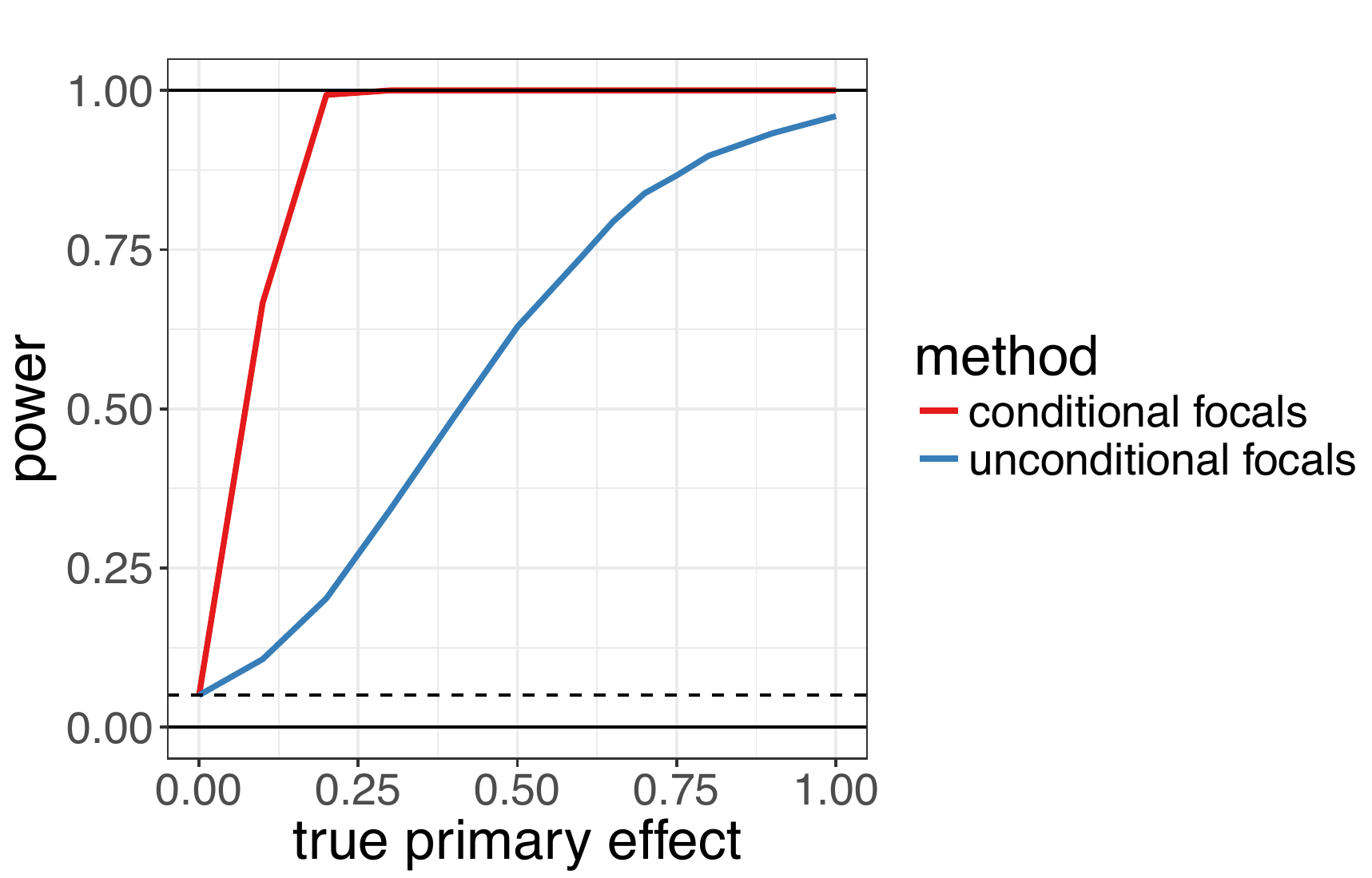}
	\caption{Power of the test of no primary effect obtained with 
    choice of focals unconditional to the observed assignment 
    (\cite{athey2016exact}) versus conditional choice, for different 
    true values of the true primary effect.}
	\label{fig:extreme}
\end{figure}

%If we are interested in testing the no spillover effect hypothesis $H_0^s$ the expected difference in 
%the number of focal units is instead:
%%
%\begin{equation}\label{eq:spillover}
%	E(K^{eff}_c - K^{eff}_r) =  K_1 / n
%\end{equation}
%

%which means that in contrast with primary effect case, in which the difference in power between the 
%the conditional focals and the random focals tests are expected to grow with the size of the households, 
%that difference tends to grow smaller as the size of the households increases. We illustrate this phenomenon
%in Figure~\ref{fig:power-by-nis}.

\begin{figure}[!t]
\centering
\begin{subfigure}{.5\textwidth}
  \centering
  \includegraphics[width=.9\linewidth]{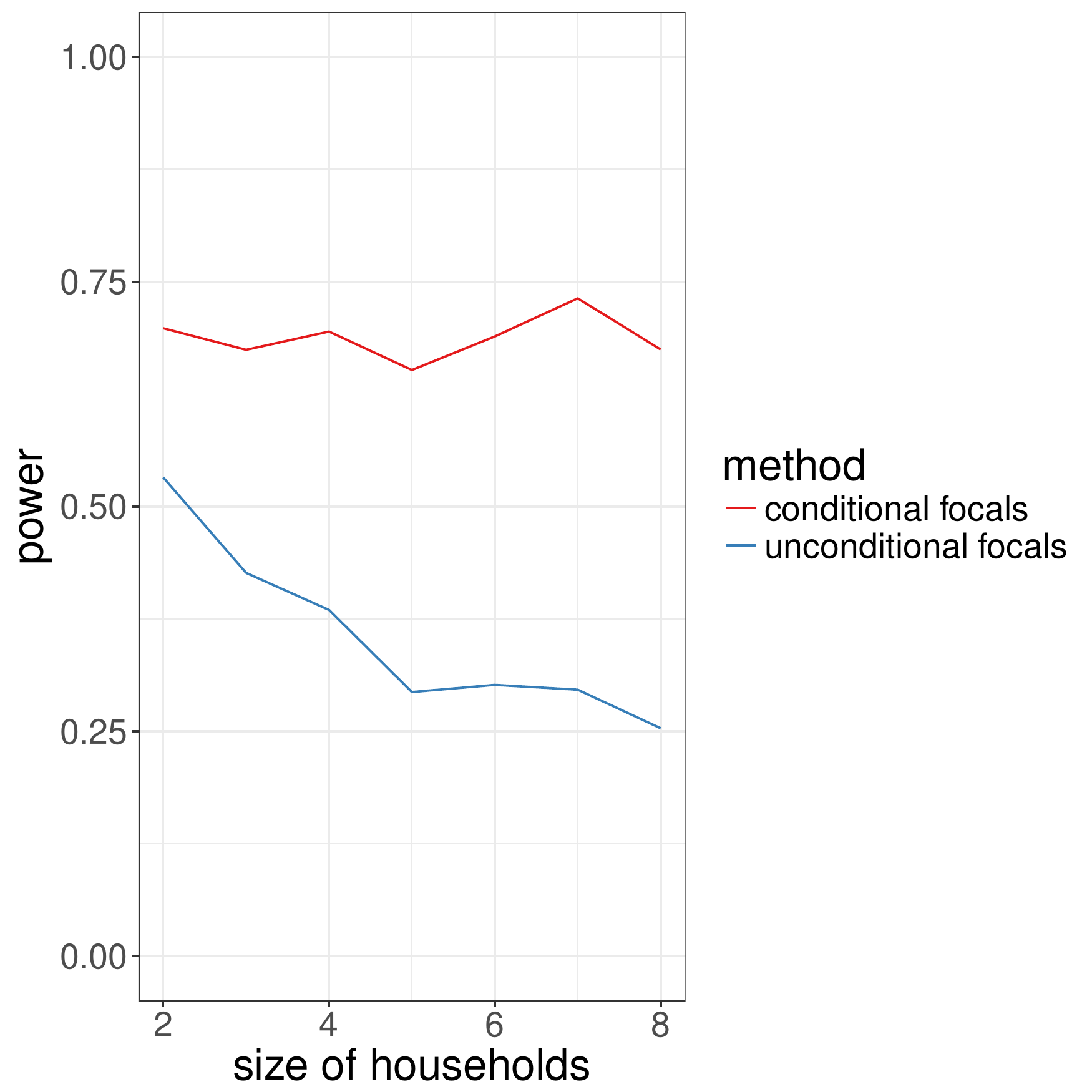}
\end{subfigure}%
\begin{subfigure}{.5\textwidth}
  \centering
  \includegraphics[width=.9\linewidth]{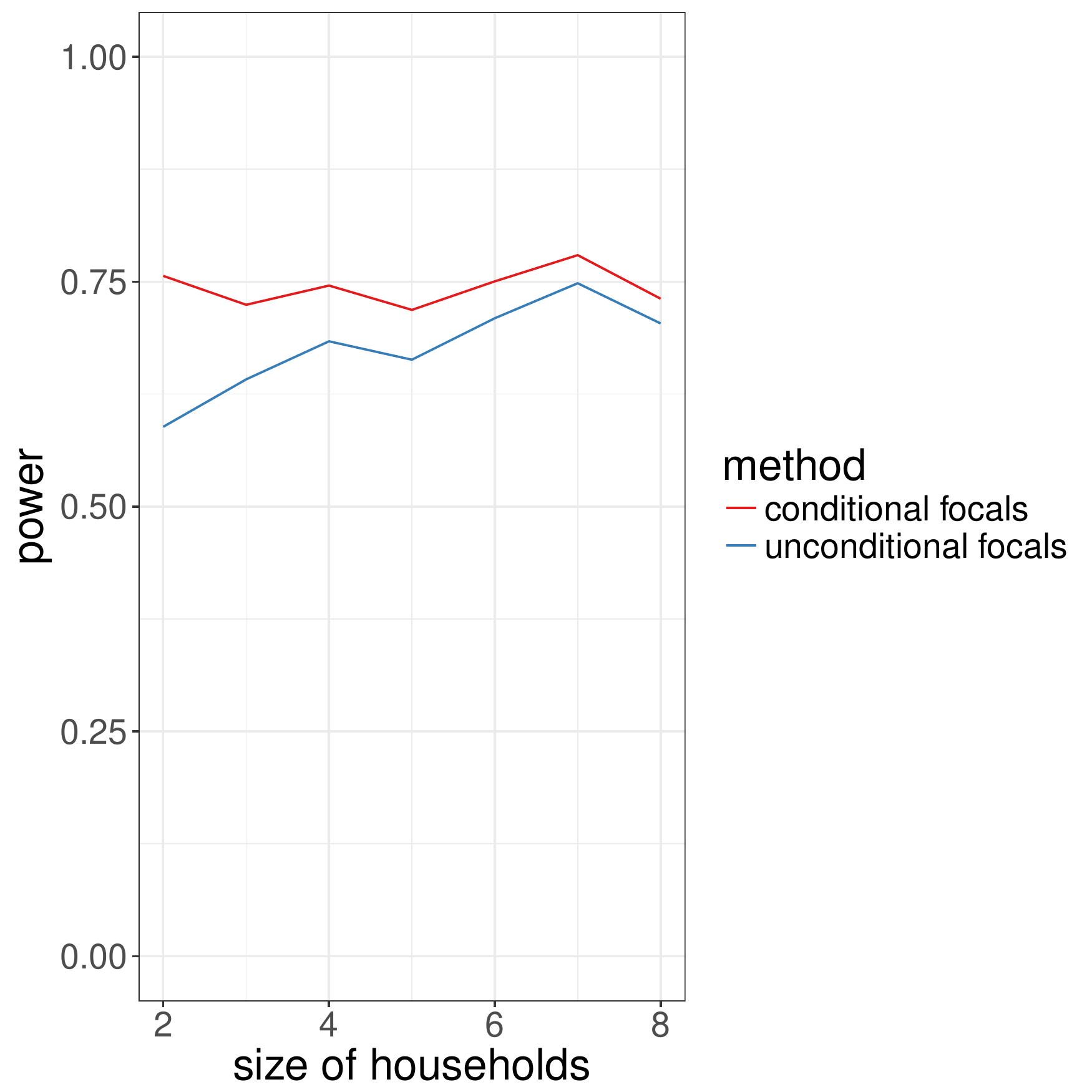}
\end{subfigure}
\caption{Power of the two methods for testing the null hypotheses of 
no primary effect, on the left, and no spillover effect, on the right, 
as a function of household size $n_i$.}
\label{fig:power-by-nis}
\end{figure}

\subsection{Details of analysis: covariate adjustment}
% - - - 

In all the analyses in the paper, covariates where taken into account 
via the same model-assisted approach used in Section~7 and Section~9.2 
of \cite{basse2016analyzing}. Briefly, we use a holdout set to estimate 
the parameter of a regression, then we use those estimators parameters to 
obtain predicted values $\{\hat{Y}_i\}_i$ for the outcomes in our sample 
and compute the residuals $\hat{e}_i = Y_i^{\mathrm{obs}} - \hat{Y}_i$. We then 
apply the conditional testing methodology to the residuals, instead of the 
original potential outcomes; in that way, the residuals can be thought of 
as transformed outcomes. Note that this approach is similar to that used by 
\cite{rosenbaum2002covariance}.

\subsection{Details of analysis: confidence intervals}
%- 
We ran an additional analysis comparing the size of confidence intervals 
for our method and for that of \cite{athey2016exact}. Specifically, for 
each of $H_0^s$ and $H_0^p$, we drew 100 focal sets using our method, and 
100 using the method of \cite{athey2016exact}, and computed the associated 
confidence intervals, obtained by inverting sequences of Fisher 
randomization tests \citep{rosenbaum2002covariance}. 
Figure~\ref{fig:ci-size} summarizes the results. We see that our method 
leads to smaller confidence intervals compared to the method of 
\cite{athey2016exact}, and that the difference is larger for the primary 
effect than for the spillover effect.
%% > END PANOS

\begin{figure}
	\centering
	\includegraphics[scale=0.4]{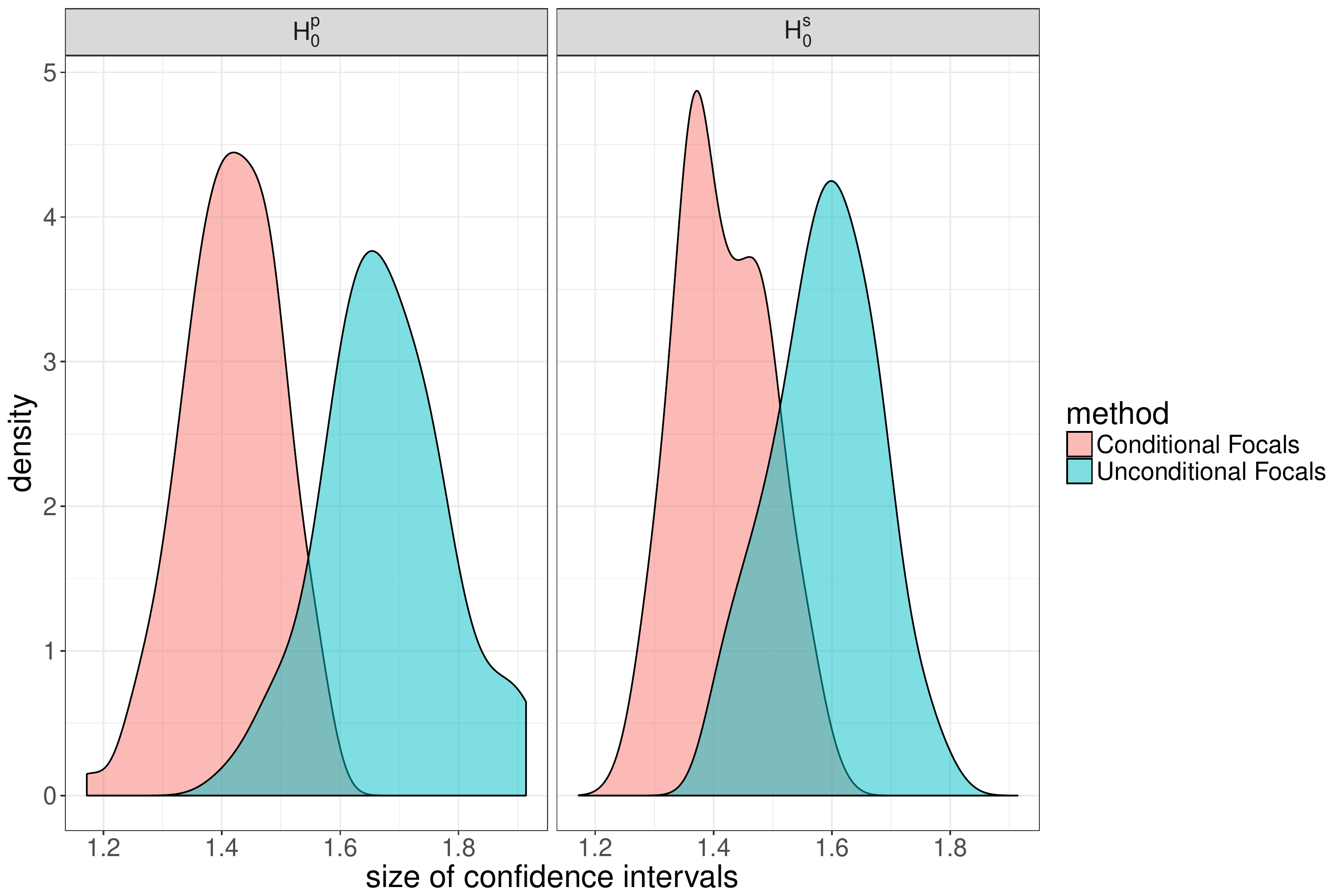}
    \caption{Size of the confidence intervals for the primary and 
    spillover effects obtained by the two methods.}
    \label{fig:ci-size}
\end{figure}

\subsection{Details of analysis: point estimates}
%-

Point estimates are obtained using a variant of the Hodges-Lehmann estimator \citep{hodges1963estimates}. Specifically, for a 
conditioning event $\C$, we numerically solve the equation $E( T \mid \C, H^P_\tau) = \Tobs$,
where $H^P_\tau$ is the null hypothesis $Y_i(1,1) = Y_i(0,0)$, by considering a grid of values for $\tau$, and 
computing the expectation of the null distribution of $T$ under the hypothesis $H^P_\tau$ and keeping the 
value $\hat{\tau}$ of $\tau$ that is closest to $\Tobs$.

\subsection{Details of analysis: results for testing $H_0^P$}
%- 

The median value of the Hodges-Lehmann for the primary effect is 
approximately equal for both choices of functions $f$ and is approximately 
equal to  $-1.5$ 
days, with associated confidence interval $[-2.2, 0.75]$ for our method, 
and $[-2.3, -0.8]$ for the method of \cite{athey2016exact}. The average 
length of confidence intervals obtained with our method is $1.4$ days, 
versus $1.6$ days for the method of \cite{athey2016exact}. The fraction of 
focals leading to a p-value below $0.05$ is $100\%$ in our case, based 
on a Monte-Carlo estimate from 100 replications, versus $92\%$ for the 
method in \cite{athey2016exact}.

% % % % % % % % % % % %% % % % % % % % % % % %% % % % % % % % % % % %

% % % % % % % % % % % % %%%%%
\section{Comparison of powers of tests}
\label{section:comparing}
% % % % % % % % % % % % % % % % 

\subsection{Model, p-values and power}
% % % % % % % % % % % % % % % %

In this section, we make an approximate theoretical analysis of the power of our test 
and the power of the test by \citet{athey2016exact}. Our analysis is 
performed under two approximations. First, in the context of classical 
Fisher randomization tests, we argue that, in general, tests that are 
balanced and use more units are more powerful. So, balance and size of 
treatment arms can be used as a proxy for the power of the test.
Second, we argue that since in the two-stage randomization case,  
our test and the test in \cite{athey2016exact} can be conceived as 
classical Fisher randomization tests run on the focal 
units, the aforementioned power 
approximation for the classical Fisher randomization test applies. 

Consider a classical Fisher randomization test, with 
complete randomization where $N_1$ out of $N$ units are treated. Let 
$p = N_1/N$. Suppose that that the true effect is constant 
additive $\tau$, and that we test for the null of no effect $H_0$. In order 
to give concrete analytical heuristics, we consider a model for the 
potential outcomes and focus on asymptotics; see also \cite{lehmann2006testing} for this approach:
\begin{equation*}
	Y_i(Z_i) \sim \tau Z_i + \mathcal{N}(\mu, \sigma^2).
\end{equation*}
As mentioned, we will focus our argument on asymptotic heuristics. Denote 
by $V = var(T \mid \Yobs, H_0)$ the randomization variance of the test 
statistics conditional on $\Yobs$, and assuming $H_0$ is true. We have, for 
large $N$:
\begin{equation*}
	V = \frac{1}{N}\bigg[ \frac{\sigma^2}{p(1-p)} + \tau^2 \bigg].
\end{equation*}
Denote by $V^{\mathrm{obs}}$ the variance of the test statistic $V^{\mathrm{obs}} = var(T)$. We have, for large $N$,
\begin{equation*}
	V^{\mathrm{obs}} = V - \frac{\tau^2}{N},
\end{equation*}
and so by applying the appropriate CLT's, we have:
\begin{equation*}
	\frac{T}{V^{1/2}} \approx \mathcal{N}(0,1), \qquad \frac{T^{\mathrm{obs}} - \tau}{(V^{\mathrm{obs}})^{1/2}} 
	\approx \mathcal{N}(0,1).
\end{equation*}
Note the application of the CLT is heuristic here, and some regularity 
conditions are required. We can then obtain an approximation of the 
distribution of a one-sided p-value for large N:
\begin{flalign*}
	\pval &= \pr(T \geq T^{\mathrm{obs}}) \\
	&\approx 1 - \Phi\bigg( \frac{T^{\mathrm{obs}}}{V^{1/2}} \bigg),
\end{flalign*}
using the asymptotics from above. We can then verify that:
\begin{flalign*}
	\frac{T^{\mathrm{obs}}}{V^{1/2}} = \frac{T^{\mathrm{obs}} - \tau}{(V^{\mathrm{obs}})^{1/2}} (1-C)^{1/2} + (NC)^{1/2} \\
	\approx W (1-C)^{1/2} + (NC)^{1/2},
\end{flalign*}
where $W \sim \mathcal{N}(0,1)$ and $C = \tau^2 [\sigma^2 / \{p(1-p)\} + \tau^2]^{-1}$ and 
so :
\begin{equation}\label{eq:pval-approx}
	\pval = 1 - \Phi(W (1 - C)^{1/2} + (NC)^{1/2}).
\end{equation}
We can use the approximation of Equation~\eqref{eq:pval-approx} to deal with the power. For $\alpha \in [0,1]$, the power of the test at level 
$\alpha$ will be
\begin{equation*}
	\beta_\alpha = \pr(\pval \leq \alpha),
\end{equation*}
but we verify that
\begin{flalign*}
	\pval \leq \alpha \quad \Leftrightarrow \quad W \geq \frac{\Phi^{-1}(1-\alpha),
	- (NC)^{1/2}}{(1-C)^{1/2}}
\end{flalign*}
and so the power of the test will be approximately:
\begin{equation}
	\beta_\alpha = 1 - \Phi\bigg( \frac{\Phi^{-1}(1-\alpha) - (NC)^{1/2}}{(1-C)^{1/2}}\bigg).
\end{equation}

\subsection{Comparing classical tests}
\label{section:comparing-tests}
% % % % % % % % % % % % % % % 

We are interested in comparing tests with different proportions $p$ of treated units, and
with different numbers $N$ of units. We will denote these quantities by $N^{(1)}$ and $N^{(2)}$
for the number of units, and $p^{(1)}$ and $p^{(2)}$ for the proportions. Let $\beta^{(1)}$ and
$\beta^{(2)}$ be the associated powers. Finally, notice that:
\begin{flalign*}
	\beta^{(1)} \leq \beta^{(2)} &\Leftrightarrow  \frac{\Phi^{-1}(1-\alpha) - 
	(N^{(1)}C^{(1)})^{1/2}}{(1-C^{(1)})^{1/2}}
	\geq \frac{\Phi^{-1}(1-\alpha) - (N^{(2)}C^{(2)})^{1/2}}{(1-C^{(2)})^{1/2}} \\
	&\Leftrightarrow \gamma^{(1)} \geq \gamma^{(2)}
\end{flalign*}
where $\gamma^{(1)} = \{\Phi^{-1}(1-\alpha) - 
(N^{(1)}C^{(1)})^{1/2}\}/\{(1-C^{(1)})^{1/2}\}$

Suppose that both tests have the same number of units 
$N^{(1)} = N^{(2)} = N$, but different fractions of treated units 
$p^{(1)} \neq p^{(2)}$. We have
\begin{flalign*}
	\gamma^{(1)} - \gamma^{(2)} &= N^{1/2} \bigg( \frac{(C^{(2)})^{1/2}}{1-(C^{(2)})^{1/2}} - 
	\frac{(C^{(1)})^{1/2}}{1 - (C^{(1)})^{1/2}}\bigg)
	 + \bigg( \frac{\Phi^{-1}(1-\alpha)}{(1-C^{(1)})^{1/2}} -  
	 \frac{\Phi^{-1}(1-\alpha)}{(1-C^{(2)})^{1/2}}\bigg) \\
	 &\rightarrow N^{1/2} \bigg( \frac{(C^{(2)})^{1/2}}{1-(C^{(2)})^{1/2}} 
	 - \frac{(C^{(1)})^{1/2}}{1 - (C^{(1)})^{1/2}}\bigg)
\end{flalign*}
and so for large $N$, 
\begin{flalign*}
	\gamma^{(1)} - \gamma^{(2)}  \geq 0 &\Leftrightarrow  \frac{(C^{(2)})^{1/2}}{1-(C^{(2)})^{1/2}} 
	- \frac{(C^{(1)})^{1/2}}{1 - (C^{(1)})^{1/2}}  \geq 0 \\
	&\Leftrightarrow p^{(1)}(1-p^{(1)}) \leq p^{(2)}(1-p^{(2)})\\
	&\Leftrightarrow |p^{(1)} - \frac{1}{2}| \geq |p^{(2)} - \frac{1}{2}|.
\end{flalign*}
So in conclusion:
\begin{equation*}
\beta^{(1)} \leq \beta^{(2)} \quad \Leftrightarrow \quad |p^{(1)} - \frac{1}{2}| \geq |p^{(2)} - \frac{1}{2}|
\end{equation*}
which, in words, means that the balanced test has more power 
asymptotically.

Suppose that $N^{(1)} \neq N^{(2)}$ but that the fractions of treated units 
in each test is identical. That is, $p^{(1)} = p^{(2)} = p$. The immediate 
consequence is that $C^{(1)} = C^{(2)} = C$, and so:
\begin{equation*}
	\gamma^{(1)} - \gamma^{(2)} 
    = \frac{C^{1/2}}{1-C^{1/2}} \bigg((N^{(2)})^{1/2} 
    - (N^{(1)})^{1/2}\bigg),
\end{equation*}
and so:
\begin{equation*}
	\beta^{(1)} \leq \beta^{(2)} \quad \Leftrightarrow N^{(1)} \leq N^{(2)},
\end{equation*}
which in words means that the test with more units has more power 
asymptotically. 

\subsection{Comparing our test with test of \cite{athey2016exact}}
% % % % % % % % % % % % % % % 

If we restrict our attention to the special case where all households have equal size $n_i = n$, then both our 
method and the method of \cite{athey2016exact} can be seen as classical Fisher randomization tests applied 
 on a set of "effective" focal units, where the set of "effective focals" is always at least as large with our 
 method as in the method of \cite{athey2016exact}, and is always balanced if the initial assignment $\pr(Z)$ is 
 balanced. We can then leverage the result of Section~\ref{section:comparing-tests} to argue heuristically that for classical Fisher randomization tests, larger and more balanced is generally better, and so we expect our method to lead to more powerful test. This has been confirmed in the simulations of Section~\ref{section:simulations} and in the analysis.

\subsection{Comparison with unconditional focal selection under a different 
design}
% % % % % % % % % % % % % % % 
In this section, we perform an analysis outside of the two-stage design setting to illustrate the generality of our framework. We assume there is a network between units such that $\mathcal{N}_i$ denotes the neighborhood of unit $i$. As in the two-stage setting, we will show that being able to condition on the observed treatment assignment, which is possible in our framework, can lead to better randomization tests.

We consider a network between units and the following exposure functions:
\begin{equation*}
	h_i(Z) = \begin{cases}
    	a &\text{ if } \quad Z_i = 1, \\
        b &\text{ if } \quad Z_i = 0, \,\, \sum_{j \in \mathcal{N}_i} Z_j < d,\\
        c &\text{ if } \quad Z_i = 0, \,\, \sum_{j \in \mathcal{N}_i} Z_j > d, \\
    \end{cases}
\end{equation*}
and assume that $N_1$ units are treated completely at random in the network, and that we 
wish to test the null hypothesis:
\begin{equation*}
	H_0: Y_i(Z) = Y_i(Z'), i=1, \ldots, N, \,\,\, \text{ for all } Z, Z': h_i(Z), h_i(Z') \in \{b,c\}.
\end{equation*}    
This example is very different from the two-stage randomization setting 
considered in the main text, but there 
is one commonality: the units who received treatment are useless for 
testing $H_0$, and so it is wasteful to include them in the focal set.
It is easy to verify that if focals are chosen completely at random, the 
distribution of the effective number of focals is 
$|\textsc{eff}(\Uset)| \sim M - \text{Hypergeom}(N, N_1, M)$, and so 
the expected number of focal units is 
$E\{|\textsc{eff}(\Uset)|\} = M - M (N_1/N)$. 
In the case where half the units are treated, that is $N_1 = N/2$, we 
have:
\begin{equation*}
	E\{|\textsc{eff}(\Uset)|\} = \frac{M}{2},
\end{equation*}
so in effect we lose half of the focal units. Choosing focals 
unconditionally but based on $\epsilon$-nets would be better than choosing 
the focals completely at random but would not solve the fundamental reason 
why power is lost. Moreover, if choosing focals based on $\epsilon$-nets 
is helpful, then it could always be combined with conditioning on the 
observed assignment to yield an even more powerful test.

To illustrate our framework in this setting, 
we could use following procedure:
\begin{enumerate}
	\item Draw $Z$, completely at random with 
    $N_1$ treated units, and $N_0$ control units.
	\item Choose $M$ focal units at random among the $N_0$ units with 
    $Z_i = 0$. Let $\Uset$ be the set of focal units. 
    \item Draw $Z' \sim \pr(Z' \mid \Uset)$ as follows. Set $Z'_i = 0$ for 
    all $i \in \Uset$. Then choose $N_0 - M$ units at random among the 
    $N - M$ non-focal units, and set $Z'_i = 0$ for these units. Finally, 
    set $Z_i = 1$ for the remaining $N_1$ units.
\end{enumerate}

We claim that the abovementioned procedure in Step 3 samples indeed from the correct conditional randomization distribution.
\begin{proof}
    By definition of the procedure in Steps 2 and 3, it holds that $\pr(Z') \propto 1 $ if $\sum_i Z_i' = N_1$,
  and also $\pr(\Uset \mid Z') = $ const., if $|\Uset| = M$ and $Z'_i=0$ for every $i\in\Uset$.
  Therefore,
    $\pr(Z' \mid \Uset) = \text{Unif}(\Zset(\Uset))$, where 
    $\Zset(\Uset) = \{Z: \text{for all } i \in \Uset, Z_i = 0 \,\, \text{ and } 
    \,\, \sum_i Z_i = N_1\}$. Which is what step 3 does.
\end{proof}

Note that in this case our approach does not lead to a permutation 
test; and neither does the method of Athey et al. Nevertheless, it leads 
to a procedure that is easily implementable and that uses more information 
than that of Athey et al.

% % % % % % % % % % % %% % % % % % % % % % % %% % % % % % % % % % % %

\section{Testing the null hypothesis of no primary effect}
% % % % % % % % % % % %% % % % % % % % % %

The paper focused on testing the null hypothesis of no spillover effects $H_0^S$. In this section, we 
briefly give equivalent results for testing the null hypothesis of no primary effect $H_0^P$. We omit 
the proofs, since they follow exactly the same outlines as the proof for $H_0^S$. A simple choice of $f$ function for testing the null hypothesis of no primary effect is
\begin{equation*}
	f(\Uset \mid Z) = \text{Unif}\{\mathbb{U}^{(P)}(Z)\},
\end{equation*}
where
\begin{equation*}
	\mathbb{U}^{(P)}(Z) = \{\Uset \in \Uall: Z_i = 1 \Rightarrow i \in \Uset, \,\,\, \text{for all } i \in \Uset \quad \text{ and } \quad \sum_i \iv(i \in \Uset) R_{ij} = 1, \text{for every household } j\}.
\end{equation*}
%
%Consider the test statistic:
%%
%\begin{equation}\label{eq:T-spill}
%	T(Z \mid Y, \C) = Ave(Y_i \mid h(Z) = (1,0) ,\, \Uset_i =1) - Ave(Y_i \mid h(Z) = (0,0),\,  \Uset_i = 1)
%\end{equation}
%%
If applied to Theorem~2, this choice of $f$ leads to the following procedure, which mirrors that of 
Proposition~1:
\begin{enumerate}
	\item In control households, $W_j = 0$, choose one unit at random. 
    In treated households, $W_j=1$, choose the treated unit as focal.
	
	\item Compute the distribution of the test statistic Equation~\eqref{eq:T-contrast} induced by all 
	 permutations of exposures on focal units, using $a = (0,0)$ and $b=(1,1)$ as 
	the contrasted exposures.
	
	\item Compute the p-value.
\end{enumerate}
This procedure is valid conditionally and marginally for testing $H_0^P$.

% % % % % % % % % % % %% % % % % % % % % %% % % % % % % % % % % %% % % % % % 

% % % % % % % % % % % %% % % % % % % % % %
\section{Notes on the choice of exposure mapping $h()$}
% % % % % % % % % % % %% % % % % % % % % %

\subsection{More complex exposure mappings}
% % % % % % % % % % % % % % % 

The class of null hypotheses that our method is designed to test is summarized in Equation~(7) of our 
manuscript, reproduced below for convenience:
\begin{equation}
\label{eq:test}
H_0: Y_i(Z) = Y_i(Z'), (i=1, \ldots, N) \,\, \text{for all} \,\, Z, Z' \text{ for which } 
h_i(Z), h_i(Z') \in \{a,b\},
\end{equation}
for some exposure function $h$, the choice of which is limited by a few theoretical and practical considerations. 
The only strong theoretical constraint implicit in Equation~(7) of the manuscript is that the two exposures 
$a$ and $b$ being contrasted must be well defined for all units under consideration. 
For instance, in the test of no spillovers $H_0^s$, the two exposures contrasted are the spillover exposure 
$(1,0)$, and the control exposure $(0,0)$, which are well defined for all units. If we had households with a 
single individual, then the exposure $(1,0)$ would not be defined for that unit and the null hypothesis of 
Equation~(7) would consequently be ill-posed if it included that unit. 

Still, the formulation in Equation~\eqref{eq:test} provides enough flexibility to test a wide variety of null hypotheses. Here, we illustrate with a couple of short but representative examples on network interference. Similar to~\cite{athey2016exact}, let $G_{ij}=1$ if units $i$ and $j$ are neighbors in the network, 
and $G_{ij}=0$ otherwise. By convention, $G_{ii}=0$ for all $i$.

Suppose we want to test spillovers on control units from first-order neighbors.
Then, we could define:
\begin{equation*}
	h_i(Z) = \begin{cases}
		a &\text{if $Z_i = 1$},\\
		b &\text{if $Z_i = 0$, \, $\sum_{j} G_{ij} Z_j > 0$},\\
		c &\text{if $Z_i = 0$, \, $\sum_{j} G_{ij} Z_j = 0$}.
	\end{cases}	
\end{equation*}
Now testing the hypothesis in Equation~\eqref{eq:test} 
contrasting the exposures $b$ and $c$ defined above will test whether there are spillovers on control units.

As another example, suppose we want to test spillovers on control units from up to second-order neighbors. Let $H_{ij} = 1$ if $i$ and $j$ are second-order neighbors but not first-order neighbors, so $G_{ij}=0$.
Then, we could define:
\begin{equation*}
	h_i(Z) = \begin{cases}
		a&\text{if $Z_i=1$},\\
		b&\text{if $Z_i=0$, \, $\sum_{j} (G_{ij} + H_{ij}) Z_j > 0$},\\
		c&\text{if $Z_i=0$, \, $\sum_{j} (G_{ij} + H_{ij}) Z_j  = 0$}.
	\end{cases}
\end{equation*}
Now testing the hypothesis in Equation~\eqref{eq:test} 
contrasting the exposures $b$ and $c$ defined above will test whether there are spillovers on control units from first-order or second-order neighbors. 
We could also test the hypothesis that there are no second-order 
spillovers without putting constraints on first-order spillovers. 
For that test, we could define:%
\begin{equation*}
	h_i(Z) = \begin{cases}
		a &\text{if $Z_i = 1$},\\
		b &\text{if $Z_i = 0$, \, $\sum_{j} G_{ij} Z_j > 0$},\\
		c &\text{if $Z_i = 0$, \, $\sum_{j} G_{ij} Z_j = 0$, \, 
		$\sum_{j} H_{ij} Z_j > 0$},\\
		d &\text{if $Z_i = 0$, \, $\sum_{j} (G_{ij} + H_{ij}) Z_j  = 0$}.
	\end{cases}
\end{equation*}
Now testing the hypothesis in Equation~\eqref{eq:test} 
contrasting the exposures $c$ and $d$ defined above will test whether there are spillovers on control units from second-order neighbors only.  We can follow similar approaches for testing higher than second-order 
spillovers. 

We now consider an example closer to the scenario of our application. Consider the same design as in 
our manuscript, but assume that all households have $n=3$ units. We are interested in testing whether 
an untreated unit in a treated household receives a different spillover if the eldest of its two siblings is treated 
compared to the spillover received if the youngest of its two siblings is treated. 

In order to test this null 
hypothesis, we need to consider a more complex exposure mapping than the one in our manuscript.
Let $E_i \in \{0,1\}$ be the treatment assignment of the eldest of unit $i$'s two siblings, and consider the exposure mapping:
$$h_i(Z) = (H_i, Z_i, E_i).$$
Each unit now has four potential outcomes:
\begin{equation*}
	Y_i(Z) \in \{ Y_i(1,1,0), Y_i(0,0,0), Y_i(1,0,1), Y_i(1,0,0)\},
\end{equation*}
the other combinations being impossible. With this exposure mapping the 
null hypothesis of no differential spillover effect from the eldest sibling can be written as:
\begin{equation*}
	H_0: Y_i(1,0,1) = Y_i(1,0,0) \,\, (i=1,\ldots, N).
\end{equation*}

\subsection{Exposure mappings and the choice of test statistic}
% % % % % % % % % % % % % % % 

The choice of test statistic $T$ is related to the choice of exposure mapping $h$ to the extent that it provides a good estimate of the 
differential effect between exposures $a$ and $b$ in Equation~\eqref{eq:test}. Furthermore, if we have 
some prior belief about the potential outcomes and the interference structure, it can be incorporated in the test statistic.~\cite{athey2016exact} have a nice and insightful discussion about possible test statistics in Section 5.3 of their paper, which is applicable in our setting as well.

\end{document}